%% file: main.tex
\documentclass[12pt]{article} 

\input{macro}

\title{\Large{\textbf{
Is Generative AI an Existential Threat to Human Creatives? Insights from Financial Economics
}}%
\footnote{I thank Zhiguo He, Richard Talley, Baozhong Yang, Liyan Yang, and Ziyu Yao for helpful discussions.}
}
\author{
Jiasun Li\footnote{\href{mailto:jli29@gmu.edu}{\tt jli29@gmu.edu}. George Mason University, 4400 University Drive, MSN 1B8, Fairfax, VA 22030, USA.}
\\
Work-in-progress \\
(Comments welcomed)
}
\date{April 7, 2024}

\begin{document}
\maketitle

\begin{abstract}
With the phenomenal rise of generative AI models (e.g., large language models such as GPT or large image models such as Diffusion), there are increasing concerns about human creatives' futures. Specifically, as generative models' power further increases, will they eventually replace all human creatives' jobs? We argue that the answer is ``no," even if existing generative AI models' capabilities reach their theoretical limit. Our theory has a close analogy to a familiar insight in financial economics on the impossibility of an informationally efficient market [\cite{grossmanImpossibilityInformationallyEfficient1980}]: 
If generative AI models can provide all the content humans need at low variable costs, then there is no incentive for humans to spend costly resources on content creation as they cannot profit from it. But if no human creates new content, then generative AI can only learn from stale information and be unable to generate up-to-date content that reflects new happenings in the physical world. This creates a paradox.
\end{abstract}
\vfill

\newpage
\doublespacing


\section{Introduction}

The past year has been drawing many people's attention to generative AI applications such as conversational AI bot ChatGPT, as well as many other transformer-based Large Language Model (LLM) applications, such as Google's \href{https://gemini.google.com/app}{Gemini} (formerly known as Bard), Meta's \href{https://llama.meta.com/}{Llama 2}, 
diffusion-based text-to-image solutions like \href{https://stability.ai/}{Stable Diffusion} or \href{https://openai.com/dall-e-3}{DALL$\cdot$E 3}, and transformer/diffusion-based video generation applications such as OpenAI's \href{https://openai.com/sora}{Sora}. All these generative applications allow a user to give a simple instruction (``\textit{prompt}"), and the AI would then quickly generate a coherent textual answer, graphical image, or video clip (``\textit{completion}"). Although the pre-training process of such large models may be prohibitively costly, once trained the models can generate new content in response to prompts with negligible costs and at a significantly faster speed than human content creators. 

The power of generative AI models has thus been argued to impose an existential threat to human content creators (including journalists, graphic designers, film directors, or even academic researchers).%
\footnote{For example, in response to Sora's release, Elon Musk (likely jokingly) \href{https://twitter.com/elonmusk/status/1758375648605745411}{tweeted} ``gg humans."}
Will generative AI models take over creative jobs, just like how automation and machines take over manual jobs (e.g., \cite{zhang2019labor})? Futuristically, as generative AI models further improve and reach their theoretical performance limit, will they completely replace human content creators? 

In this paper, we address this very question of whether generative AI models could completely replace human content creators. We give a resounding ``no" answer based on a simple intuition borrowed from a familiar insight in financial economics: Specifically, recall \cite{grossmanImpossibilityInformationallyEfficient1980} has articulated that a strong-form efficient market as defined in \cite{fama1976foundations} is impossible to obtain, and we argue that a close analogy exists between the impossibility of a fully informationally-efficient market and a fully generative AI dominated creatives market. The analogy is so close that even ChatGPT can ``see" it, as to be demonstrated in the following paragraph. 

To (1) succinctly summarize our idea, (2) illustrate the power of generative AI models like GPT, and (3) demonstrate how natural our proposed analogy is, we hereby present an ``AI-generated" explanation instead of using our own words:

\begin{tcolorbox}[width=\textwidth,boxsep=2pt,left=0pt,right=5pt,top=5pt,bottom=5pt]
\textbf{Prompt:} 
Does the insight in Grossman and Stiglitz (1980) AER apply to the relationship between human content creators and generative AI models like GPT?
\end{tcolorbox}

\begin{quote}
\textbf{Completion:}%
\footnote{Note that ChatGPT does not necessarily produce the truth -- therefore the reasoning here is only suggestive and will be formalized and proved in the rest of our paper. Also, note that ChatGPT is not entirely deterministic -- completions may vary slightly from time to time even for the same prompt, and there is no obvious way to set random seeds. That said, based on our experiments the variations across trials are non-material to the completion's contents.}
\textit{Grossman and Stiglitz's seminal paper in 1980, titled "On the Impossibility of Informationally Efficient Markets," puts forth the idea that perfectly efficient markets are an impossibility if the cost of information gathering is not zero. Their reasoning is simple but profound: if prices fully reflect all available information, then there's no incentive for traders to spend resources to acquire information since they cannot profit from it. But if no one acquires information, prices can't be informationally efficient. This creates a paradox.}

\textit{Drawing an analogy to the context of human content creators and generative AI models like GPT:}

\begin{enumerate} 
    \item \textit{\textbf{Information Acquisition Cost:} In content creation, the effort and cost involved in humans researching, understanding, and then creating content can be likened to the cost of information acquisition in the Grossman-Stiglitz framework. GPT-like models can generate content at a fraction of the time and potentially cost, making them akin to a reduced 'information acquisition cost' mechanism in this analogy.}

    \item \textit{\textbf{Efficiency \& Incentives:} If AI models can generate vast amounts of high-quality content instantaneously, the "market" (here, the content landscape) could be saturated with this content, making it challenging for human creators to compete or find a niche. The efficiency of AI in content generation might reduce the incentive for human creators to produce content, similar to how perfect informational efficiency reduces the incentive for traders to acquire information.}

    \item \textit{\textbf{Value of Originality \& Authenticity:} Just as there's value in the original research and information gathering in financial markets (which gives an edge to traders), there's inherent value in the originality, creativity, and authenticity of human content creators. This value might not be entirely replicable by AI models, no matter how advanced.}

    \item \textit{\textbf{Paradoxical Equilibrium:} Extending the Grossman-Stiglitz insight, if everyone relies solely on AI for content, there might be a homogenization of content, reducing diversity, depth, and potentially quality. This could open spaces for human creators to bring in unique perspectives, just as there will always be a space for traders who do the hard research work in financial markets.}

    \item \textit{\textbf{Continuous Feedback Loop:} An interesting dimension is the feedback loop where human creators could use AI models to assist in their content creation process, just as traders use algorithms and models to assist in trading but don't solely rely on them.}
\end{enumerate}
    
\textit{In essence, while the Grossman-Stiglitz (1980) paper specifically addressed financial markets, its insights about the balance between efficiency, incentives, and the cost of acquiring information can be drawn upon analogously when thinking about the evolving dynamics between human content creators and advanced generative AI models.}
\end{quote}

As ChatGPT has clearly articulated (based on parroting our existing knowledge of the \cite{grossmanImpossibilityInformationallyEfficient1980} model and the inner workings of generative AI models), the interaction between generative AI and human creatives are similar to those between active and passive investors in the financial markets. Therefore, while the rise of generative AI models may take over some human content creation jobs, they will eventually reach an interior equilibrium in which generative AI and human creatives coexist -- just like how active and passive investors coexist in today's financial market. 

As far as we know, this paper is the first attempt to theoretically analyze the long-term societal impact of generative AI using insights from financial economics. Besides the paper's own implications, we also hope to demonstrate the power of classic ideas of financial economics for deepening our understanding of the rising technologies in AI. While the rise of generative AI may be a new phenomenon, the power and relevance of profound financial economics insights do not turn old.

\paragraph{Our contributions}

Our contributions is two-fold. 
First, it analytically answers a question of significant contemporaneous community interest and utmost long-term societal importance. 
Second, we again demonstrate the importance of interdisciplinary thinking for approaching generative AI topics, and especially the relevance of decades-old insights from financial economics for the emerging area.

\paragraph{Related literature}

Straddling two distinct areas of generative AI and financial economics, this paper contributes to multiple strands of the literature. 

First, our results explore the societal impacts of generative AI models, which are typically developed in the machine-learning community. Notable contributions to the recent success of generative AI models include ones on transformer architecture such as \cite{vaswani2017attention} and \cite{scao2022bloom}, pre-training and scaling laws such as \cite{kaplan2020scaling}, model architectures and pre-training objectives such as \cite{wang2022language}, \cite{touvron2023llama}, and scaling laws and compute-optimal models such as \cite{brown2020language}, \cite{hoffmann2022training}, and \cite{wu2023bloomberggpt}, as well as the opportunities and risks of foundation models \cite{bommasani2021opportunities}. While this paper's focus is not on the technical details of these models, our theory crucially builds on the fundamental properties of AI models distilled from production-level models. 


Second, our results also contribute to a small but rapidly growing literature studying the potential social and economic impacts of generative AI models. Several papers have examined AI's disruption of labor markets, including \cite{eisfeldt2023generative} on how generative AI could affect firm values, \cite{horton2023large} viewing large language models as simulated economic agents, and \cite{brynjolfsson2023generative} looking broadly at generative AI's impact on the future of work. Other works focus on policy implications, such as \cite{korinek2022preparing} on preparing for AI's workforce effects, \cite{eloundou2023gpts} analyzing language models' labor market potential, and \cite{acemoglu2023regulating} discussing regulating transformative AI technologies. From a macro perspective, \cite{johnson2023power} explore AI's role in economic prosperity over history. Specific to the financial industry,  
\cite{grennan2020artificial} look at the implication of AI proliferation on financial analyst careers, while \cite{jha2024chatgpt} investigate ChatGPT's capability in forecasting corporate decisions. Compared to these models, our paper features a unique perspective from the information economics of financial markets. To the extent AI can be viewed as a form of advance in data technology, our results also contribute to an emerging literature on the implications of data technology advances, such as \cite{farboodi2020long}, \cite{berg2020rise}, \cite{bryzgalova2023bayesian}, \cite{abis2024changing}. 

There is also an emerging literature that applies generative AI models to the asset market itself. \cite{lopez2023can} show that 
applying 
ChatGPT and other large language models (LLMs) 
to news headlines can 
significantly predict subsequent daily stock returns. 
\cite{Chen2024ChatGPT} find that good news extracted by ChatGPT from the front pages of Wall Street Journal can predict the stock market and is related to macroeconomic conditions. 
Similarly, \cite{tan2023large} utilize large language models (LLMs) to extract contextualized representations of Chinese news text and predict stock returns in Chinese equity market. 
All these papers mainly focus on empirically evaluating the power of generative models in the equity market, while our focus is to theoretically analyze the interaction between human and AI content creators, although our analysis also has implications that can directly apply to the stock market as well. 

Finally, our results also infuse new implications into classic theories, especially the rational expectation equilibrium (REE) models for the financial markets, such as \cite{grossmanEfficiencyCompetitiveStock1976}, \cite{grossmanImpossibilityInformationallyEfficient1980}, \cite{hellwigAggregationInformationCompetitive1980}, \cite{diamondInformationAggregationNoisy1981}, 
\cite{admatiDirectIndirectSale1990}, \cite{kyleInformedSpeculationImperfect1989}, \cite{bondRealEffectsFinancial2012}, \cite{breon-drishExistenceUniquenessEquilibrium2015}, and \cite{li2023information}, etc. Our dynamic model also generalizes existing attempts at the dynamics of information models, such as \cite{wang1993model}. 
Our paper clarifies how the insights of financial economists accumulated over the past several decades will continue to survive and thrive in the era of generative AI. 

\section{Backgrounds}

To appreciate how our theory accurately captures the reality associated with generative AI models, it is imperative to adequately understand how generative AI models function. Since these models are still relatively new, and their inner workings may not be immediately familiar to non-AI researchers (such as financial economists), Section \ref{sec_how_genAI_works} first provides the necessary technical background. 

On the other hand, to appreciate the real-world implications of our analysis, it is also helpful to realize that the fear of generative AI replacing human creatives is not purely theoretical. To this end, Section \ref{sec_news} quotes recent news articles reporting human creatives' responses to such concerns as motivational examples. 

\subsection{A crash course on generative AI models} \label{sec_how_genAI_works}

In this section, we provide a high-level introduction to how generative AI models work. For exposition ease, we will use the GPT model as an example. 

Since this part of the paper is not our novel contribution (although a self-contained explanation of how generative AI models work is important for understanding the rest of the paper), to remain unbiased, we again leverage the power of ChatGPT to answer the prompt of ``how does ChatGPT work?"
For readers' ease, we will \underline{underscore} the key takeaway within the ChatGPT's answer (``completion"), which appears at the end of this completion. Other than the \underline{underscored} part crucial for the understanding of the rest of the paper, impatient readers may safely skip the rest of this subsection without loss of continuity. 

\begin{tcolorbox}[width=\textwidth,boxsep=2pt,left=0pt,right=5pt,top=5pt,bottom=5pt]
\textbf{Prompt:} 
How does ChatGPT work?
\end{tcolorbox}

\begin{quote}
\textbf{Completion:} 
\textit{Certainly! At a high level, ChatGPT (like the one you're speaking with) is based on the GPT (Generative Pre-trained Transformer) architecture, which is a type of machine learning model specifically designed for natural language understanding and generation.}

\textit{Here's a breakdown of how ChatGPT works:}

\begin{enumerate}
    \item \textit{\textbf{Architecture:} GPT is based on the Transformer architecture, which uses attention mechanisms to weight input data differently, allowing the model to focus on different parts of the input when generating an output. This ability to ``pay attention" to certain parts of the input is part of what makes GPT models especially adept at understanding context in sequences of data, like sentences or paragraphs.}

    \item \textit{\textbf{Training:} GPT models undergo a two-step training process:
    \begin{enumerate}
        \item \textbf{Pre-training:} The model is trained on vast amounts of text data (like books, websites, etc.) to predict the next word in a sentence. This helps the model learn grammar, facts about the world, reasoning abilities, and some level of commonsense knowledge. It's a form of unsupervised learning because the model isn't given explicit "answers"; instead, it learns patterns from the data.
        \item \textbf{Fine-tuning:} After pre-training, the model is further trained on more specific tasks with labeled data. For instance, to make ChatGPT, the base GPT model is fine-tuned on data that includes a chat-like format to improve its ability to have conversations.
    \end{enumerate}}
    \item \textit{\textbf{Inference:} Once the model is trained, you can provide it with a prompt (like the questions you're typing), and it will generate a response based on the patterns it has learned. It doesn't ``understand" the text in the way humans do, but it generates outputs based on the massive amount of text it was trained on.}
    \item \textit{\textbf{Parameters:} GPT models, especially the latest versions like GPT-4, have billions or even trillions of parameters. These parameters are tweaked during training to help the model represent and generate text. The sheer number of parameters is part of what allows the model to store and generate such a vast range of information.}
    \item \textit{\textbf{Limitations:} While ChatGPT can generate coherent and contextually relevant text, it doesn't truly ``understand" the content. It might produce incorrect or nonsensical answers, can be sensitive to the phrasing of the input, doesn't know specifics about events after its last training cut-off, and doesn't have personal experiences or emotions.}
\end{enumerate}

\textit{To visualize the process, \ul{imagine you're trying to teach a parrot to converse like a human. If you expose the parrot to vast amounts of human conversation, it will start to mimic phrases and structures it hears often. But even if the parrot gets really good at this mimicry, it doesn't truly understand what it's saying. GPT-4 is a very sophisticated version of this "parrot," with a vast memory of phrases and structures to draw upon.}}
\end{quote}

A key takeaway from the above \ul{underscored} paragraph is that generative AI models do not create ``new content," but rather piggyback on existing content. For example, GPT as well as other Large Language Models merely approximate the joint distribution of existing human languages and spit out the most likely ``completion" conditional on the prompt it receives. Similar logic also applies to other image-based or video-based applications. 

More formally, suppose $F$ denotes the joint distribution of all possible meaningful content that the human race may ever encounter (the population), one can reasonably view generative AI as using what humans have created in the past, or more precisely, the training corpse included during these model's training processes (the sample), to infer an approximating distribution $\hat{F}$, and then for any given prompt $X$ generating the most likely completion $Y$ by maximizing $\hat{F}(Y|X)$. All generative AI models follow this high-level framework, with differences in their intended ideal distribution $F$ (e.g., either for text or image) and the adopted algorithms for learning $\hat{F}$. 

Although our explanation has skipped most of the technical details within the GPT training algorithm, including important components such as input/position encoding, token embedding, transformer (encoder/decoder blocks and attention), etc., it will be clear that our model does not rely on the actual implementations of GPT models (that is, the specific algorithms used for learning $\hat{F}$). Therefore, our theory will apply to all existing generative AI models, all of which feature the ``parrot" property described above.  

\subsection{Anecdotal evidence} \label{sec_news}

The concern over generative AI replacing human creatives is not merely an academic theorization, but is rather an imminent reality. The following excerpt from a SeekingAlpha article on August 14th, 2023 gives a vivid example: 

\begin{quote}
    \textit{With generative AI poised to reshape the publishing industry, companies are scrambling to take positions on the latest technology. What does it mean for their operations? How will it impact their workforce? And when can they begin to expect to benefit from, or be harmed by, the new developments? These questions have rattled publishers in nearly every sector of the media industry, from online content and the printed word to movies and even music.}

    \textit{Snapshot: The knee-jerk reaction of many firms was to band together to stave off any threats to their bottom line. Hollywood actors and writers have gone on strike in part over fear that they might lose or share revenue with machines, with the walkout now \href{https://seekingalpha.com/news/4000719-writers-strike-100-days-no-end-on-the-horizon?hasComeFromMpArticle=false&lctg=60c1e9011abeb364a710717b&mailingid=32390809&messageid=wall_street_breakfast&serial=32390809.1058371&userid=3879781&utm_campaign=wsb_08_14_2023&utm_medium=email&utm_source=seeking_alpha&utm_term=wall_street_breakfast}{going on for more than 100 days}.
    Publishers like The New York Times (NYT), Wall Street Journal's News Corp. (NWS) and Dotdash Meredith owner IAC (IAC) also started talks to establish a united front against Google (GOOG) (GOOGL) and Microsoft (MSFT). The coalition explored suing Big Tech over "content scraping" and \href{https://seekingalpha.com/news/3989502-ai-google-sergey-brin-project-gemini}{training their AI models} on data produced by the organizations, but the NYT has reportedly dropped out of the group, suggesting that they might reach separate deals with tech giants that wield tremendous power over distribution and visibility. Check out recent subscriber numbers for the NYT and News Corp. in charts.}

    \textit{While corporations are debating their cash flow streams, the non-profit sector is having an easier time making decisions. In July, the Associated Press reached a two-year deal with OpenAI, the parent company of ChatGPT, by declaring the first official news-sharing agreement with an artificial intelligence firm. That direction could benefit others as they scale their offerings and revenue, such as News Corp Australia (NWS), which is using a team of four staff to produce 3,000 articles per week using generative AI. The items cover daily topics like local weather, fuel prices, and traffic conditions, and have driven a surge in subscriptions and the subsidiary's best year in a decade.}

    \textit{Outlook: Consumers are also debating what artificial intelligence means for their media consumption. Should they be informed whether a robot wrote an article, or is the creative process continuing to evolve in a similar fashion to what algo-generated tunes and CGI did to the music and movie industries? Things get even more complicated if generative AI is only used to conduct research or help write an article, or prompt topics that enhance value and quality. Trust also plays a role here, but those lines might be lost or could one day be strengthened, depending on the future of brands and their audiences.}
\end{quote} 

Such concerns over AI potentially taking over humans also go beyond journalism. For example, in the realm of business, generative AI models have also been claimed to potentially take over humans in generating business ideas. For example, in a recent experiment, \cite{girotra2023ideas} compare innovative business ideas submitted in response to the prompt asking to ``generate an idea for a new product or service appealing to college students that could be made available for \$50 or less" by both Wharton MBA students enrolled in an innovation course and GPT-4, and show that generative AI models not only produce ideas at a faster pace but also of higher quality (as tested by putting the ideas in front of a survey panel of customers in the target market via an online purchase-intent survey).%
\footnote{That said, doubts remain whether the AI-generated ideas are truly innovative, or rather than just rehashing existing ideas. For example, in a \href{https://www.reddit.com/r/singularity/comments/16fy0a7/mba_students_vs_chatgpt_who_comes_up_with_more/}{Reddit discussion} over a Wall Street Journal article "\href{https://www.wsj.com/tech/ai/mba-students-vs-chatgpt-innovation-679edf3b}{M.B.A. Students vs. AI: Who Comes Up With More Innovative Ideas?}" based on the research, commentators have pointed out that (WARNING: explicit language and typos copied verbatim, may be offensive to some readers): 
\begin{quote}
    \textit{I can't believe this has so many upvotes, although I guess maybe it's not that surprising if you don't actually read the article and the study itself. I personally, respectfully, think this is fucking dumb as shit. Here are the "top ideas" from GPT:}
    \begin{itemize}
        \item \textit{Compact Printer}
        \item \textit{Solar-Powered Charger}
        \item \textit{Mini Vacuum}
        \item \textit{Noise Cancelling headphones}
    \end{itemize}
    \textit{... Notice something? This shit already fucking exists, and is successful in the market. It's already popular. And the study design was to then take the ideas and ask people if they'd purchase it.}

    \textit{The rest are pretty much the same.}

    \textit{A first aid kit?}

    \textit{A portable lamp?}

    \textit{Are you fucking serious?}

\textit{Someone took this and wrote an article about "innovate ideas"? And how ChatGPT came up with more "innovative ideas"?}
\end{quote} 
The fact that generative AI tends to ``regenerate" existing ideas is consistent with our theory, as generative AI models are clever ``parrots" that regurgitates human-created contents.}

\section{A Reduced-form Model }


\label{section_static_model}

We first build a simple static model to illustrate our core idea. The model in this subsection is intentionally kept to be highly reduced-form to capture the essence of how the low-cost ``parrot" nature of generative AI prevents it from completely replacing human creatives. 
Appendix \ref{section_dynamic_model} sketches a more microfounded model to capture the dynamic interaction between human creatives and generative AI in a more full-fledged model. 

\paragraph{Agents}
There is a unit continuum of consumers for creative content. Among them, a fraction \( \lambda \) of agents are human creatives, who can incur a cost $C>0$ to create new content for their own consumption (for simplicity, we do not consider the possibility of hiring other human creatives for content creation --- or alternatively, we can view the current model as treating content creators and consumers as the same type of agents) while a fraction \( 1-\lambda \) of agents only use generative AI with a lower cost $c<C$ for their content needs. 
In subsequent discussions, we may simply call the second type of agent generative AI for brevity. Note that the fraction $\lambda$ will be endogenously determined in equilibrium. 

\paragraph{Preferences}
Normalizing a consumer's utility when consuming no content (either from human creation or generative AI) to zero, we denote the benefit from consuming human research as $R$. We assume that $R>C$, as otherwise, the world will trivially end up in an uninteresting as well as unrealistic situation of no content creation/consumption. We denote the payoff to an agent who conducts human content creation as $ \pi_{\text{human}} = R - C > 0$.

On the other hand, we denote the benefit from consuming generative AI-created research as $r$. Therefore, the payoff to an agent who uses generative AI $ \pi_{\text{AI}} = r(\lambda) - c$. Here, we write $r(\lambda)$ to highlight that $r$ is a function of $\lambda$, the extent to which human creatives produce content. More formally, we hold the following assumptions for the function $r$: 
\begin{assumption} \label{r_assumption}
    The benefit from consuming generative AI-created content, $r(\lambda)$, is an increasing function of $\lambda$, the extent to which human creatives produce content. 
\end{assumption}
Assumption \ref{r_assumption} naturally follows how generative AI models work as we have previously explained in Section \ref{sec_how_genAI_works}. To see this, notice that generative AIs are clever ``parrots" that learn from human-generated content. Keeping the algorithms of a generative AI model fixed, the performance of the generative AI model depends on how comprehensive and up-to-date its training sample is. To the extent that more content created by human creatives (that is, a higher $\lambda$) helps make the generative AI model's training sample more comprehensive and up-to-date, which in turn improves its performance, the benefit it brings to content consumers $r(\lambda)$ also increases with the intensity of human content creation $\lambda$. 

\begin{assumption} \label{r_assumption2}
    $r(0) = 0$ and $r(1) \approx R$. 
\end{assumption}

Assumption \ref{r_assumption2} also naturally follows how generative AI models operate as previously explained in Section \ref{sec_how_genAI_works}. 
To see this, notice that when $\lambda=0$, no content is created by humans so that the AI will have no input for its training purpose and therefore not be able to produce anything meaningful. We thus have $r(0) = 0$. 
Regarding the value of $r(1)$, we consider two cases: (1) If the AI model is perfect, in the sense that the AI is assumed to learn everything that has been created by humans so far, then we would have $r(1) = R$. Otherwise, (2) if the AI is imperfect, then $r(1) < R$. 
Since we are more interested in the case in which generative AI's capability reaches its theoretical limit,  
we shall assume that $r(1) = R - \epsilon$, where $\epsilon$ is some infinitesimal amount. Combining the two cases, we assume that $r(1) \approx R$. 

\paragraph{Equilibrium characterization}

We now proceed to solve for the equilibrium outcomes of the game, characterized by a couple of lemmas that lead to the main proposition. 

\begin{lemma} \label{thm_no_full_AI}
    Generative AI completely taking over human creatives' jobs cannot be an equilibrium outcome. In other words, $\lambda > 0$ in equilibrium.
\end{lemma}
\begin{proof}
    We prove the lemma by contradiction. Suppose otherwise, then $\lambda = 0$, and each generative AI user gets $r(0) - c = -c < 0$. If one of them were to deviate toward human content creation, then his payoff would become $R - C > 0$, which is profitable. Therefore, $\lambda = 0$ cannot hold in equilibrium.
\end{proof}
Lemma \ref{thm_no_full_AI} addresses the main question posted at the beginning of the paper regarding whether generative AI will completely take over human content creation -- the answer is a resounding no. 

However, we do expect generative AI to play significant roles in equilibrium, as the next lemma indicates. 
\begin{lemma} \label{thm_no_full_human}
    It is not an equilibrium for generative AI to not be adopted, i.e., in equilibrium $\lambda < 1$. 
\end{lemma}
\begin{proof}
    We again prove the lemma by contraction. Suppose otherwise, then $\lambda = 1$, and each human content creator gets $R - C$. If he were to deviate toward using generative AI, his payoff would be $r(1) - c = R - \epsilon - c > R - C$. Therefore, $\lambda = 1$ cannot hold in equilibrium. 
\end{proof}

Given Lemma \ref{thm_no_full_AI} and \ref{thm_no_full_human}, we have that the equilibrium must be interior, that is, generative AI and human content creators must coexist. Therefore, in equilibrium we must have $ \pi_{\text{human}} = \pi_{\text{AI}} $, or $R-C = r(\lambda)-c$. 
\begin{proposition}
    There exists a unique $\lambda^*$ satisfying $R-C = r(\lambda^*)-c$, which characterize the unique equilibrium in our model economy
\end{proposition}
\begin{proof}
    Notice that the right-hand side of equation $R-C = r(\lambda)-c$ (with respect to $\lambda$) is monotonically increasing, is lower than the left-hand side when $\lambda=0$, and is higher than the left-hand side when $\lambda=1$. 
\end{proof}

From the definition of $\lambda^*$, we can immediately obtain the following comparative statics: 
The equilibrium level of the fraction of human content creators, $\lambda^*$, increases in the value of human-generated content $R$, decreases in the cost of producing human-generated content $C$, decreases in the marginal cost of AI content generation $c$, and decreases in the performance of generative AI. All these comparative statistics should be intuitive.

\section{Conclusion}

In this paper, we make a simple point, that generative AI will not completely take over human content creators. The core idea shares a lot of resemblance with the active-passive tension in traditional financial markets. Our result serves to demonstrate that some fundamental ideas in financial economics could still be highly relevant as we enter an era of AI. 

There are many extensions of our core model that future research may further explore. For example, we so far assume that content consumers all create content for themselves, while in reality, they may hire other parties (e.g., professional creatives) for their content needs; Future research may follow the framework in \cite{admatiDirectIndirectSale1990} to further tease out the equilibrium outcome. We have also so far ruled out potential heterogeneity in all the content created; Future work may leverage insights from \cite{hellwigAggregationInformationCompetitive1980} or \cite{diamondInformationAggregationNoisy1981} to fully develop the case with heterogeneous contents. Our model so far has assumed only one type of content desired by the population; When the interactions among multiple content markets (e.g., general text, general image, or domain-specific contents) are considered, one may adopt ideas from \cite{admatiNoisyRationalExpectations1985}. Furthermore, to the extent that content consumers may act upon what they learn, which in turn affects the overall informational landscape, one may further explore the feedback loop between human-generated content, AI-generated content, and the real distribution of information in the physical world, leveraging insights from \cite{chen2007price, edmans2015feedback, goldstein2017information}, etc. 
Finally, we may further derive dynamic interactions between human creatives and generative AI, highlighting the fact that generative AI models are trained on lagged samples, and their values decay if no additional input are generated by human content creators. 


A key friction behind why generative AI cannot completely replace humans is that AI cannot learn all input humans can perceive. In other words, existing generative models need to rely on the first-hand content created by humans to generate second-hand information. The conclusion of our paper may change, however, if this fundamental friction is no longer present. One foreseeable way this friction may disappear is when generative AIs can directly learn from the physical world. The fact that many AI models are now increasingly multimodal or even omni-model, and that sensor technologies further improve (including human-machine connections as recent progress from Neuralink promises), then our conclusions may need to be qualified -- however, that would be a discussion of certain more advanced AI paradigms that we have not yet seen today. 

\bibliographystyle{aea}
\bibliography{main}

\end{document}

%% file: macro.tex
\usepackage{tcolorbox}
\usepackage{soul}

\usepackage{natbib} 
\usepackage{graphicx} 
\usepackage{amsmath}
\usepackage{comment}
\usepackage{booktabs} 
\usepackage{amsthm,mathtools,amssymb,mathrsfs,bbm}
\usepackage{longtable,tabularx}
\usepackage{setspace}
\usepackage{fullpage}
\usepackage{tabu,adjustbox,multirow,hhline,makecell,rotating,changepage,textcomp,color,subfigure,fancyhdr,pdflscape,dsfont,caption,dcolumn}

\definecolor{webblue}{rgb}{0,0,0.6}

\usepackage[pagebackref,breaklinks]{hyperref}
\hypersetup{colorlinks=true, anchorcolor= webblue, citecolor= webblue, filecolor= webblue, linkcolor= webblue, menucolor= webblue, urlcolor= webblue, citebordercolor= 1 0 0, menubordercolor=1 0 0, urlbordercolor=1 0 0, runbordercolor=1 0 0} 

\newtheorem{theorem}{Theorem}
\newtheorem{lemma}[theorem]{Lemma}
\newtheorem{proposition}[theorem]{Proposition}

\newtheorem{assumption}{Assumption}

\usepackage{tikz} 
\usetikzlibrary{decorations.pathmorphing} 
\usetikzlibrary{decorations.pathreplacing} 
\usetikzlibrary{arrows}
\usepackage{pifont}
\usetikzlibrary{decorations.shapes}
\usetikzlibrary{decorations}
\usetikzlibrary{shapes.arrows}
\usepackage{pgfkeys}
\usepackage{pdflscape}
\usepackage{lscape}
\usetikzlibrary{backgrounds}
\allowdisplaybreaks

\usetikzlibrary{decorations.text}

\usetikzlibrary{arrows,calc}
\tikzset{>=stealth', help lines/.style={dashed, thick}, axis/.style={<->}, important line/.style={thick}, connection/.style={thick, dotted},
}